\makeatletter
\def\input@path{{styles/}}
\makeatother

\ifx\GenSoCGVer\undefined%
\documentclass[12pt]{article}%
\newcommand{\SoCG}[1]{}
\newcommand{\NotSoCG}[1]{#1}%
\else
\newcommand{\SoCG}[1]{#1}
\newcommand{\NotSoCG}[1]{}%
\documentclass[a4paper,UKenglish,anonymous,thm-restate]%
{socg-lipics-v2021} \fi

\NotSoCG{%
   \def\UseBibLatex{1}%
}

\ifx\sarielComp\undefined%
\newcommand{\SarielComp}[1]{}
\newcommand{\NotSarielComp}[1]{#1}%
\else
\newcommand{\SarielComp}[1]{#1}%
\newcommand{\NotSarielComp}[1]{}%
\fi

\usepackage{hyperref}%

\hypersetup{%
   unicode,
      breaklinks,%
      colorlinks=true,%
      urlcolor=[rgb]{0.25,0.0,0.0},%
      linkcolor=[rgb]{0.5,0.0,0.0},%
      citecolor=[rgb]{0,0.2,0.445},%
      filecolor=[rgb]{0,0,0.4},
      anchorcolor=[rgb]={0.0,0.1,0.2}%
}
\usepackage[ocgcolorlinks]{ocgx2}
\usepackage{memhfixc}

\NotSoCG{%
   \usepackage[in]{fullpage}%
}

\usepackage{amsmath}%
\usepackage{amssymb}%
\usepackage{xcolor}%

\SarielComp{\usepackage{sariel_colors}}%

\NotSoCG{%
   \usepackage[amsmath,thmmarks]{ntheorem}%
   \theoremseparator{.}%

   \usepackage{titlesec}%
   \titlelabel{\thetitle. }%
}

\usepackage{xcolor}%
\usepackage{mleftright}%
\usepackage{xspace}%
\usepackage{graphicx}

\NotSoCG{%
\theoremseparator{.}%

\theoremstyle{plain}%
\newtheorem{theorem}{Theorem}[section]

\newtheorem{lemma}[theorem]{Lemma}

\newtheorem{corollary}[theorem]{Corollary}
\newtheorem{claim}[theorem]{Claim}%
\newtheorem{fact}[theorem]{Fact}

\theoremstyle{plain}%
\theoremheaderfont{\sf} \theorembodyfont{\upshape}%
\newtheorem*{remark:unnumbered}[theorem]{Remark}%
\newtheorem{remark}[theorem]{Remark}%
\newtheorem{definition}[theorem]{Definition}
\newtheorem{defn}[theorem]{Definition}

\newcommand{\myqedsymbol}{\rule{2mm}{2mm}}

\theoremheaderfont{\em}%
\theorembodyfont{\upshape}%
\theoremstyle{nonumberplain}%
\theoremseparator{}%
\theoremsymbol{\myqedsymbol}%
\newtheorem{proof}{Proof:}%

}

\SoCG{%
   \theoremstyle{remark}
   
   \newtheorem{fact}[theorem]{Fact}
}

\definecolor{blue25emph}{rgb}{0, 0, 11}
\providecommand{\emphic}[2]{%
   \textcolor{blue25emph}{%
      \textbf{\emph{#1}}}%
   \index{#2}}

\providecommand{\emphi}[1]{\emphic{#1}{#1}}

\definecolor{almostblack}{rgb}{0, 0, 0.3}
\providecommand{\emphw}[1]{{\textcolor{almostblack}{\emph{#1}}}}%

\newcommand{\SarielThanks}[1]{%
   \thanks{%
      Department of Computer Science; %
      University of Illinois; %
      201 N. Goodwin Avenue; %
      Urbana, IL, 61801, USA; %
      \href{mailto:spam@illinois.edu}{sariel@illinois.edu}; %
      \url{http://sarielhp.org/}. %
   #1%
   }%
}
\newcommand{\MariaThanks}[1]{%
   \thanks{%
      Department of Computer Science; %
      University of Illinois; %
      201 N. Goodwin Avenue; %
      Urbana, IL, 61801, USA; %
      \href{mailto:spam@illinois.edu}{marial5@illinois.edu}. %
   #1%
   }%
}

\newcommand{\HLink}[2]{\hyperref[#2]{#1~\ref*{#2}}}
\newcommand{\HLinkY}[2]{\hyperref[#2]{#1}}
\newcommand{\HLinkSuffix}[3]{\hyperref[#2]{#1\ref*{#2}{#3}}}

\newcommand{\thmlab}[1]{{\label{theo:#1}}}
\newcommand{\thmref}[1]{\HLink{Theorem}{theo:#1}}
\newcommand{\thmrefY}[2]{\HLinkY{#2}{theo:#1}}

\newcommand{\Xthmref}[1]{\noexpand{\noexpand\HLink{Theorem}{theo:#1}}}%

\newcommand{\corlab}[1]{\label{cor:#1}}
\newcommand{\corref}[1]{\HLink{Corollary}{cor:#1}}%

\newcommand{\seclab}[1]{\label{sec:#1}}
\newcommand{\secref}[1]{\HLink{Section}{sec:#1}}

\newcommand{\clmlab}[1]{\label{claim:#1}}
\newcommand{\clmref}[1]{\HLink{Claim}{claim:#1}}

\newcommand{\itemlab}[1]{\label{item:#1}}
\newcommand{\itemref}[1]{\HLinkSuffix{}{item:#1}{}}

\newcommand{\factlab}[1]{\label{fact:#1}}%
\newcommand{\factref}[1]{\HLink{Fact}{fact:#1}}%

\newcommand{\remlab}[1]{\label{rem:#1}}
\newcommand{\remref}[1]{\HLink{Remark}{rem:#1}}%

\newcommand{\lemlab}[1]{\label{lemma:#1}}
\newcommand{\lemref}[1]{\HLink{Lemma}{lemma:#1}}%
\newcommand{\Xlemref}[1]{\noexpand{\noexpand\HLink{Lemma}{lemma:#1}}}%

\providecommand{\deflab}[1]{\label{def:#1}}
\renewcommand{\deflab}[1]{\label{def:#1}}
\newcommand{\defref}[1]{\HLink{Definit\-ion}{def:#1}}
\newcommand{\defrefY}[2]{\hyperref[def:#2]{#1}}

\newcommand{\obsrefY}[2]{\hyperref[observation:#1]{#2}}

\providecommand{\eqlab}[1]{}%
\renewcommand{\eqlab}[1]{\label{eq:#1}}
\newcommand{\Eqlab}[1]{\label{eq:#1}}
\newcommand{\Eqref}[1]{\HLinkSuffix{Eq.~(}{eq:#1}{)}}

\newcommand{\remove}[1]{}%
\newcommand{\Set}[2]{\left\{ #1 \;\middle\vert\; #2 \right\}}
\newcommand{\pth}[2][\!]{\mleft({#2}\mright)}%
\newcommand{\pbrcx}[1]{\left[ {#1} \right]}%
\newcommand{\Prob}[1]{\mathop{\mathbb{P}}\!\pbrcx{#1}}
\newcommand{\Ex}[2][\!]{\mathop{\mathbb{E}}#1\pbrcx{#2}}

\newcommand{\ExSym}{\mathop{\ExChar}}%
\newcommand{\ExChar}{\mathbb{E}}%
\newcommand{\ExCond}[2]{\ExSym\!\left[%
       #1 \;\middle\vert\; #2 \right]}
\newcommand{\ProbLTR}{\mathbb{P}}%
\newcommand{\ProbCond}[2]{\mathop{\ProbLTR}\!\left[%
       #1 \;\middle\vert\; #2 \right]}

\usepackage{stmaryrd}%
\providecommand{\IntRange}[1]{\mleft\llbracket #1 \mright\rrbracket}
\newcommand{\IRX}[1]{\IntRange{#1}}%
\newcommand{\IRY}[2]{\left\llbracket #1:#2 \right\rrbracket}

\newcommand{\VV}{\Mh{\mathsf{V}}}%
\newcommand{\VX}[1]{\VV\pth{#1}}%

\providecommand{\EG}{\Mh{\mathsf{E}}}%
\providecommand{\EGX}[1]{\EG\pth{#1}}%

\newcommand{\ceil}[1]{\left\lceil {#1} \right\rceil}
\newcommand{\floor}[1]{\left\lfloor {#1} \right\rfloor}

\newcommand{\cardin}[1]{\left| {#1} \right|}%

\renewcommand{\th}{th\xspace}

\renewcommand{\Re}{\mathbb{R}}%

\usepackage[inline]{enumitem}

\newlist{compactenumA}{enumerate}{5}%
\setlist[compactenumA]{topsep=0pt,itemsep=-1ex,partopsep=1ex,parsep=1ex,%
   label=(\Alph*)}%

\newlist{compactenuma}{enumerate}{5}%
\setlist[compactenuma]{topsep=0pt,itemsep=-1ex,partopsep=1ex,parsep=1ex,%
   label=(\alph*)}%

\newlist{compactenumI}{enumerate}{5}%
\setlist[compactenumI]{topsep=0pt,itemsep=-1ex,partopsep=1ex,parsep=1ex,%
   label=(\Roman*)}%

\newlist{compactenumi}{enumerate}{5}%
\setlist[compactenumi]{topsep=0pt,itemsep=-1ex,partopsep=1ex,parsep=1ex,%
   label=(\roman*)}%

\newlist{compactitem}{itemize}{5}%
\setlist[compactitem]{topsep=0pt,itemsep=-1ex,partopsep=1ex,parsep=1ex,%
   label=\ensuremath{\bullet}}%

\newcommand{\UsePackage}[1]{%
  \IfFileExists{styles/#1.sty}{%
      \usepackage{styles/#1}%
   }{%
      \IfFileExists{../styles/#1.sty}{%
         \usepackage{../styles/#1}%
      }{%
         \usepackage{#1}%
      }%
   }%
}

\providecommand{\BibLatexMode}[1]{}
\providecommand{\BibTexMode}[1]{#1}

\ifx\UseBibLatex\undefined%
  \renewcommand{\BibLatexMode}[1]{}
  \renewcommand{\BibTexMode}[1]{#1}
\else
  \renewcommand{\BibLatexMode}[1]{#1}
  \renewcommand{\BibTexMode}[1]{}
\fi

\BibLatexMode{%
   \usepackage[bibencoding=utf8,style=alphabetic,backend=biber]{biblatex}%
   \UsePackage{sariel_biblatex}%
}

\numberwithin{figure}{section}%
\numberwithin{table}{section}%
\numberwithin{equation}{section}%

\providecommand{\Mh}[1]{#1}%

\newcommand{\M}{\Mh{\mathcal{M}}}

\newcommand{\G}{\Mh{\mathsf{G}}}
\renewcommand{\H}{\mathsf{H}}

\providecommand{\K}{\mathsf{K}}
\renewcommand{\K}{\mathsf{K}}

\newcommand{\Event}{\mathcal{E}}%

\usepackage{todonotes}%
\usepackage{mathcalb}%

\newcommand{\nFailX}[1]{\mathcalb{f}\pth{#1}}
\newcommand{\pdX}[1]{\Mh{\zeta}\pth{#1}}
\newcommand{\gC}{\Mh{\mathcalb{g}}}%
\newcommand{\gX}[1]{\gC_{#1}}%
\newcommand{\goodX}[1]{\gC\pth{#1}}%
\newcommand{\goodY}[2]{\gC\pth{#1, #2}}%
\newcommand{\Rch}{\Mh{\mathcal{R}}}%
\newcommand{\RchX}[1]{\Rch_{#1}}%
\newcommand{\prevX}[1]{\mathrm{prev}\pth{#1}}

\newcommand{\dGX}[1]{\Mh{\mathcalb{d}}_{#1}}
\newcommand{\dGZ}[3]{\Mh{\dGX{#1}}\pth{#2,#3}}
\newcommand{\dGY}[2]{\Mh{\mathcalb{d}}\pth{#1,#2}}

\newcommand{\DistrC}{\mathcal{D}}%
\newcommand{\DistrY}[2]{\DistrC\pth{#1,#2}}

\newcommand{\lbC}{\Mh{\mathcalb{\ell}}}
\newcommand{\lbY}[2]{\lbC\pth{#1,#2}}
\newcommand{\lbZ}[3]{\lbC_{\leq #3}\pth{#1,#2}}
\providecommand{\TPDF}[2]{\texorpdfstring{#1}{#2}}

\newcommand{\cODS}{c_6}
\newcommand{\cKHop}{c_7}
\newcommand{\sL}{\Mh{\mathcalb{l}}}%
\newcommand{\sB}{\Mh{\mathcalb{b}}}%
\newcommand{\BSet}{\Mh{\mathcal{B}}}%
\newcommand{\set}[1]{\left\{ {#1} \right\}}

\newcommand{\pr}{\varrho}
\renewcommand{\P}{\Mh{P}}%
\newcommand{\eps}{\Mh{\varepsilon}}%
\newcommand{\kk}{\Mh{\mathsf{k}}}%
\newcommand{\nbl}{\Mh{\nabla}}%

\newcommand{\pA}{\Mh{u}}
\newcommand{\pB}{\Mh{v}}
\newcommand{\pC}{\Mh{x}}%

\newcommand{\p}{\pA}
\newcommand{\q}{\pB}

\newcommand{\UC}{[0,1)^d}

\newcommand{\order}{\Mh{\sigma}}%
\newcommand{\HC}{\Mh{\mathcal{H}}}%

\newcommand{\ball}{\Mh{\mathcalb{b}}}

\newcommand{\ballY}[2]{\ball\pth{#1, #2}}
\newcommand{\normX}[1]{\left\| {#1} \right\|}
\newcommand{\dY}[2]{\normX{#1 #2}}
\newcommand{\etal}{\textit{et~al.}\xspace}
\newcommand{\LSO}{\Term{LSO}\xspace}%
\newcommand{\Term}[1]{\textsf{#1}}
\newcommand{\Eps}{\Mh{\mathcal{E}}}%
\newcommand{\epsA}{\eps}%
\newcommand{\orderset}{\Pi}
\newcommand{\ordAll}{\orderset^+}%
\newcommand{\logeps}{\log(1/\eps)}

\newcommand{\BadSet}{\Mh{{B}}}%
\newcommand{\EBadSet}{\Mh{{B^+}}}%
\newcommand{\NN}{\mathbb{N}}%
\newcommand{\epsR}{\Mh{\vartheta}}%

\usepackage{scalerel}%
\newcommand{\tldO}{\scalerel*{\widetilde{O}}{j^2}}%

\newcommand{\dependable}{dependable\xspace}

\newcommand{\prS}{\Mh{\psi}}
\newcommand{\si}[1]{#1}

\providecommand{\tcite}[1]{\cite{#1}}

\newcommand{\sBVal}{\ceil{ \frac{\cKHop \nbl}{\prS} \ln n}}

\newcommand{\sBValS}{\ceil{ \smash{\frac{\cKHop \nbl}{\prS} } \ln n}}
\newcommand{\GGen}{\widehat{\G}}

\SoCG{\def\ProofInAppendix{TRUE}}%
\usepackage{proofapnd}%

\SoCG{%
   \newcommand{\myparagraph}[1]{\subparagraph{#1}}
}%
\NotSoCG{%
   \newcommand{\myparagraph}[1]{\paragraph*{#1}}
}%

\BibLatexMode{%
   \bibliography{dependable_spanner}%
}

\begin{document}

\title{Dependable Spanners via Unreliable Edges}%
\NotSoCG{%
   \author{%
      Sariel Har-Peled%
      \SarielThanks{Work on this paper was partially supported by NSF
         AF awards CCF-1907400 and CCF-2317241. %
      }%
      \and%
      Maria C. Lusardi%
      \MariaThanks{}%
   }%
}%

\SoCG{%
   \author{Sariel Har-Peled}%
   {Department of Computer Science, University of Illinois, 201
      N. Goodwin Avenue, Urbana, IL 61801, USA %
      \and
      \url{https://sarielhp.org}%
   }%
   {sariel@illinois.edu}%
   {https://orcid.org/0000-0003-2638-9635}%
   {Work on this paper was partially supported by a NSF AF award
      CCF-2317241.}%

   \Copyright{Sariel Har-Peled and Maria C. Lusardi}

   \ccsdesc[500]{Theory of computation~Computational geometry}%

   \keywords{Spanners}%

   \category{} %

   \relatedversion{}

}

\SoCG{%
   \authorrunning{S. Har-Peled and M. C. Lusardi}%
}%
\date{\today}

\maketitle

\begin{abstract}
    Let $\P$ be a set of $n$ points in $\Re^d$, and let
    $\eps,\prS \in (0,1)$ be parameters. Here, we consider the task of
    constructing a $(1+\eps)$-spanner for $\P$, where every edge might
    fail (independently) with probability $1-\prS$. For example, for
    $\prS=0.1$, about $90\%$ of the edges of the graph
    fail. Nevertheless, we show how to construct a spanner that
    survives such a catastrophe with near linear number of edges.

    The measure of reliability of the graph constructed is how many
    pairs of vertices lose $(1+\eps)$-connectivity. Surprisingly,
    despite the spanner constructed being of near linear size, the
    number of failed pairs is close to the number of failed pairs if
    the underlying graph was a clique.

    Specifically, we show how to construct such an exact dependable
    spanner in one dimension of size $O(\tfrac{n}{\prS} \log n)$,
    which is optimal.  Next, we build an $(1+\eps)$-spanners for a set
    $\P \subseteq \Re^d$ of $n$ points, of size $O( C n \log n )$,
    where $C \approx 1/\bigl(\eps^{d} \prS^{4/3}\bigr)$. Surprisingly,
    these new spanners also have the property that almost all pairs of
    vertices have a $\leq 4$-hop paths between them realizing this
    short path.
\end{abstract}

\section{Introduction}

\subsection{Background}

\myparagraph{Spanners.}%

Given a weighted finite graph $\M$ over a set of points $\P$ (if $\M$
is a finite metric, then $\M$ is a clique), a \emphw{$t$-spanner} is a
subgraph $\G \subseteq \M$, such that for all $\pA,\pB \in \P$, we
have that
$\dGZ{\M}{\pA}{\pB} \leq \dGZ{\G}{\pA}{\pB} \leq t \cdot
\dGZ{\M}{\pA}{\pB}$, where $\dGX{\M}$ and $\dGX{\G}$ denote the
shortest path length in $\M$ and $\G$, respectively.  The
\emph{weight} of an edge $\pA\pB \in \EGX{\G}$ is
$\dGZ{\M}{\pA}{\pB}$. A lot of work went into designing and
constructing spanners with various properties.  The main goal in
spanner constructions is to have small \emph{size}, that is, to use as
few edges as possible. Other properties include low degrees
\cite{abcgh-sggsd-08,cc-srgs-10,s-gsfed-06}, low weight
\cite{bcfms-cgsnq-10,gln-fgacs-02}, low diameter
\cite{ams-rdags-94,ams-dagss-99}, or resistance to failures. See
\cite{ns-gsn-07}.

\myparagraph{Fault tolerant spanners.}  A desired property of
spanner is fault tolerance
\cite{lns-eacft-98,lns-iacft-02,l-nrftg-99}. A graph $\G=(\P,\EG)$ is
an \emph{$r$-fault tolerant $t$-spanner} if for any set $\BadSet$ of
failed vertices with $\cardin{\BadSet} \leq r$, the graph
$\G \setminus \BadSet$ is still a $t$-spanner.  The disadvantage of
$r$-fault tolerance is that each vertex must have degree at least
$r+1$, otherwise the vertex can be isolated by deleting its
neighbors. Therefore, the graph has size at least $\Omega(rn)$. In
particular, for $r$ large the size of the fault-tolerant spanner is
prohibitive.

\myparagraph{Region fault tolerant spanners.}  Abam \etal
\cite{adfg-rftgs-09} showed that one can build a geometric spanner
with near linear number of edges, so that if the deleted set are all
the points belonging to a convex region (they also delete the edges
intersecting this region), then the residual graph is still a spanner
for the remaining points\footnote{More precisely, the remaining graph,
   as edges passing through the ``bad'' region are also deleted.}.

\myparagraph{Vertex robustness.}
For a function~$f:\NN \xrightarrow{} \Re^+$ a $t$-spanner $\G$ is
$f$-robust \cite{bdms-rgs-13}, if for any set of failed points
$\BadSet$ there is an extended set $\EBadSet$ (that contains
$\BadSet$) with size at most $f\pth{\cardin{\BadSet}}$ such that the
residual graph $\G\setminus\BadSet$ has a $t$-path for any pair of
points~$\p,\q \in \P \setminus \EBadSet$. The function $f$ controls
the robustness of the graph -- the slower the function grows the more
robust the graph is.  For~$\epsR\in(0,1)$, a spanner that is
$f$-robust with~$f(k)=(1+\epsR)k$ is a \emphi{$\epsR$-reliable}
spanner~\cite{bho-sda-20}.

\myparagraph{Reliable spanners for unreliable vertices.} %
Buchin \etal \cite{bho-sda-20} showed a construction of reliable exact
spanners of size $\tldO\pth{n\log n}$ in one dimension, and of
reliable $(1+\eps)$-spanners of size
$\tldO\pth{n\log n \log\!\log^6\!n }$ in higher dimensions (the
constant in the $\tldO$ depends on the dimension, $\eps$, and the
reliability parameter $\epsR$). An alternative construction, with
slightly worse bounds, was given by Bose \etal \cite{bcdm-norgm-18}.
Up to polynomial factors in $\log \log n$, this matches a lower bound
of Bose \etal \cite{bdms-rgs-13}.  Buchin \etal \cite{bho-srsal-22}
showed that the size of the construction can be improved to
$\tldO(n \log\log^3 n \log \log \log n)$ if the attacker choices
(i.e., the failed set of vertices) is oblivious to the randomized
construction of the spanner. Some of these constructions use \LSO{}s,
described next.

\myparagraph{Locality sensitive orderings.}

The concept of \emphw{locality-sensitive orderings} (\LSO) was
introduced by Chan \etal \cite{chj-lota-20}. Informally, they showed
that $\Re^d$ can be multi-embedded into the real line, such that
distances are roughly preserved.
\begin{definition}
    \deflab{local}%
    For a pair of points $\p,\q \in \HC = \UC$, an order $\order$ over
    the points of $\HC$ is \emphi{$\eps$-local}, for $\eps \in (0,1)$,
    if
    \begin{equation*}
        \order(\p,\q)
        \quad\subseteq\quad
        \ballY{\p}{\eps \ell}
        \cup
        \ballY{\q}{\eps \ell} ,
        \qquad\text{where} \qquad%
        \ell = \dY{\p}{\q},
    \end{equation*}
    where $\ballY{\p}{r}$ denotes the \emphi{ball} of radius $r$
    centered at $\p$.
\end{definition}
Namely, all the points between $\p$ and $\q$ in $\order$ are in the
vicinities of $\p$ and $\q$ in $\HC$.

Surprisingly, Chan \etal \cite{chj-lota-20} showed that one
can compute such a ``universal'' set of orderings $\Pi$ (i.e., a set
of \emphw{locality-sensitive orderings}), of size
$O(\Eps^d \log \Eps)$, where $\Eps = 1/\eps$. This set of orderings
can be easily computed, and computing the order between any two points
according to a specified order (in the set) can be done quickly.
Using \LSO{}s some problems in $d$ dimensions, are reduced to a
collection of problems in one dimension.  Recently, Gao and Har-Peled
\cite{gh-nolso-24} showed improved construction of \LSO{}s of size
$O(\Eps^{d-1} \log \Eps)$, but unfortunately, these \LSO{}s have
slightly weaker properties.

\subsection{Our results}

If one is interested in building near linear size spanners, that
survive massive edge failure (i.e., constant fraction of the edges),
then this seems hopeless. Indeed, one can easily isolate (completely)
a large fraction of the vertices of the graph by deleting the edges
attached to them. One can interpret an edge failure as the failure of
both its endpoints, and use one of the constructions of reliable
spanners mentioned above, but this seems wasteful -- it assigns an
edge failure the same status as a failure of two vertices (i.e., one
can simulate the failure of an edge $uv$, as the failure of one of its
endpoints $u$ or $v$), which seems excessive. In addition, this model
can withstand only (roughly) $n$ failures till the graph effectively
disappears as all the vertices are ``deleted''.

Here, we initiate the study of how to construct such spanners that can
survive massive edge failure.  Specifically, we imagine that given the
constructed spanner graph $\G$, and a parameter $\prS \in (0,1)$, the
edge failures are random and independent. Specifically, an edge fail
with probability $1-\prS$ (i.e., selected to the ``surviving'' graph
with probability $\prS$). Our measure of the quality of $\G$, is how
many pairs of vertices in $\G$ lose their spanning property in the
residual graph. We refer to a graph that can survive such an attack
and have ``few'' failing pairs as being a \emphw{dependable spanner},
as to distinguish this concept from the reliable spanners discussed
above (which handle vertex failures).

There are other natural choices for the underlying model -- for
example, make the failure probability of an edge a function of its
length. This in turn opens a Pandora box of further choices -- What
this function should be for the problem to make sense? Should it
depends on the local density, or location of its nodes? Etc.  To keep
things manageable, we chose the above simpler settings, and leave such
questions for further research.

\myparagraph{On the optimal deficiency in one dimension.}

The natural starting point is the complete graph $K_n$ over
$\IRX{n} = \{ 1,\ldots, n\}$. Let $\lbY{n}{\prS}$ denote the expected
number of pairs of points in $\IRX{n}$ that no longer have a straight
path in a graph sampled $\H \sim \DistrY{K_n}{\prS}$ (i.e., such a
failed pair $ i < j$ has the property that the shortest path in $\H$
between $i$ and $j$ is longer than $j-i$). The quantity
$\lbY{n}{\prS}$ is the \emphw{optimal deficiency}, and it provides a
lower bound on the number of such failed pairs in any construction.

In \secref{opt:deficiency}, we study $\lbY{n}{\prS}$.  A relatively
straightforward upper bound of $O(n/\prS^2)$ on the deficiency is
provided is \secref{u:b:easy}. To improve the upper bound, we
introduce the concept of a block -- the idea is to consider a
consecutive interval of vertices, and how many reachable vertices
there are in such a block for a fixed source. One can then argue that
the number of reachable vertices between two consecutive blocks,
behaves similarly to what expansion guarantees in a bipartite
expander.  This leads to an improved upper bound
$\lbY{n}{\prS} = O( \tfrac{n}{\prS} \log \tfrac{1}{\prS})$. See
\secref{u:b:optimal} for details. A surprisingly simple argument then
shows that there is a matching lower bound
$\lbY{n}{\prS} = \Omega( \tfrac{n}{\prS} \log \tfrac{1}{\prS})$, see
\lemref{l:b:clique}. Thus, for the optimal deficiency, we have
$\lbY{n}{\prS} = \Theta( \tfrac{n}{\prS} \log \tfrac{1}{\prS})$, see
\thmref{clique}.

\myparagraph{Constructing one dimensional dependable exact spanners.}

Equipped with the \si{above} bounds on the optimal deficiency, we
prove a lower bound on the number of edges such a graph must have --
specifically, in \secref{l:b:graph}, we show that a dependable exact
spanner on $\IRX{n}$ must have $\Omega(\tfrac{n}{\prS} \log n)$ edges,
if the deficiency is to be linear in $n$. A construction of dependable
spanner (matching this lower bound) is natural -- one connects all
vertices that are in distance $O(\tfrac{1}{\prS} \log n )$ from each
other along the line. It is not hard to show that this graph has
deficiency that is at most one bigger than the optimal, see
\lemref{long:paths}.

\myparagraph{Constructing one dimensional dependable exact spanners with
   few hops.}

For our application of building dependable spanners in higher
dimensions, we need spanners that have few hops (i.e., for almost all
pairs there is a straight path with at most $4$ edges). This turns out
to be doable, by building a $4$-hop spanner on the blocks, and then
replacing each block-edge by a bipartite clique. In the resulting
graph, the number of edges increases to $O_\prS(n \log^2 n)$, see
\clmref{spanner}.

To reduce the number of edges, we replace every bipartite clique with
a random bipartite graph (i.e., a random bipartite expander). This
results in a graph with $O(\tfrac{n}{\prS^{4/3}} \log n )$ edges, that
has the desired $4$-hop property. The proof of correctness is a bit
more subtle, as one needs to carefully argue about the underlying
expansion. One can improve the dependency on the number of
hops. Specifically, \thmref{k:hop} shows that one can construct a
spanner $\G$ with $O( (n/\prS^{1+1/(\kk-1)}) \log n)$ edges, such that
(in expectation) at most $O( n/\prS^{1+1/(\kk-1)} \log (1/\prS) )$
pairs are not connected via a $\kk$ hop path in a graph
$\H \sim \DistrY{\G}{\prS}$.

\myparagraph{Constructing dependable \TPDF{$(1+\eps)$}{1+eps}-spanners
   in \TPDF{$\Re^d$}{Rd}.}

The last step of converting these one dimensional dependable spanners
to dependable spanners in higher dimensions is by now standard. Given
a point set $\P$, we plug-in the above construction of dependable
one-dimensional spanners into the \LSO{}s provided by the construction
of Chan \etal \cite{chj-lota-20}.  Specifically, given a set $\P$ of
$n$ points in $\Re^d$, we show (see \thmref{main}) a construction of a
graph $\G$ with $C n \log n$ edges, such that in expectation at most
$C n$ pairs of points fail to be $(1+\eps)$-spanned in the randomly
sampled graph $\H \sim \DistrY{\G}{\prS}$, where each edge survives
with probability $\prS$, and $C \approx O(\eps^{-d}
p^{-4/3})$. Significantly, all the well spanned pairs are connected
via short $4$-hop paths.

\myparagraph{Similarity to reliable spanners.}

As in the work of Buchin \etal \cite{bho-sda-20} on reliable spanners,
we first solve the one dimensional version of the problem -- this is
the main contribution of this work. Once we resolve the one
dimensional problem, the \LSO{} construction of Chan \etal provides a
black-box solution to the higher dimensional version of the problem.
However, this is not the case for Buchin \etal \cite{bho-sda-20} --
their bound on the length of the generated path relied on a somewhat
involved refinement argument that does not make sense here.

The conceptual difference between the two problems is that while
reliable spanners do not ``mind'' long paths, our version of the
problem can not easily handle such paths in the same way (as longer
paths have higher probability of failure). In particular, our
construction ultimately requires the existence of ``many'' short-hop
paths between vertices, to provide the desired dependability -- this
also makes the analysis here simpler than Buchin \etal
\cite{bho-sda-20} in the $d\geq 2$ case.  Informally, the main
challenge in this work is quantifying the (expected) amount of damage
caused by the (random) attack, while in the reliable spanners settings
the challenge is on ``bypassing'' such damage.

In summary, while some underlying basic initial ideas are similar
(e.g., use bipartite expanders and \LSO{}s), the different measures
used in the two problems makes them fundamentally different.

\subsection*{Paper organization}%
We start with basic definitions in \secref{prelims}.  We study the
optimal deficiency of the clique over $\IRX{n}$ in
\secref{opt:deficiency}.  We present the one dimensional dependable
spanner in \secref{1:dim}.  We modify this construction to have few
hops in \secref{1:dim:few:hops}.  The final step of constructing the
dependable spanner in $\Re^d$ is presented in \secref{main}.  Some
open problems are described in \secref{conclusions}.

\section{Preliminaries}
\seclab{prelims}

For two integers $\alpha \leq \beta$, let
$\IRY{\alpha}{\beta} = \set{ \alpha, \alpha+1,\ldots, \beta}$.  Let
$\IRX{n} = \IRY{1}{n} = \{ 1,\ldots, n\}$, and let $K_n$ denote the
complete graph over $\IRX{n}$. Consider a subgraph $\G \subseteq
K_n$. An edge $ij \in \EGX{\G}$ has weight $|i-j|$. Let
$\dGY{i}{j} =\dGZ{\G}{i}{j}$ denote the length of the shortest path in
$\G$ from $i$ to $j$. A graph $\G \subseteq K_n$ is an \emphi{exact
   spanner} if for all $i,j \in \IRX{n}$ we gave that
$\dGY{i}{j} = |i-j|$.

\begin{definition}
    \deflab{b:clique}%
    For two disjoint sets $X,Y$, let
    $X \otimes Y = (X \cup Y, \Set{xy}{x \in X, y \in Y } )$ denote
    the \emphw{bipartite clique} on $X \cup Y$.
\end{definition}

\begin{definition}
    \deflab{filter}%
    Given a graph $\G = (\VV,\EG)$, and a parameter $\prS \in [0,1]$,
    let $\H = (\VV, \EG')$ be a subgraph of $\G$, where an edge
    $e \in \EG$ is included in $\EG'$ (independently) with probability
    $\prS$. Let $\DistrY{\G}{\prS}$ denote the resulting distribution
    over graphs. In particular, let
    $\DistrY{n}{\prS} = \DistrY{K_n}{\prS}$.
\end{definition}

The distribution $\DistrY{n}{\prS}$ is usually denoted by $G(n,\prS)$
in the literature.

\begin{definition}
    A path $\pi = i_1 i_2 \ldots i_k$ is a \emphi{straight path}
    between $i$ and $j$ in $G$, if $\pi$ is a valid path in $\G$,
    $i_1 = i$, $i_k = j$, and $i_1 < i_2 < \cdots< i_k$. It is a
    \emphi{$t$-hop} path, if $k \leq t$.
\end{definition}

Thus $\G$ is an \emphi{exact spanner} if there is a straight path in
it for all pairs of vertices\NotSoCG{ in $\IRX{n}$}.

\begin{definition}
    \deflab{deficiency}%
    For a graph $\G$ over $\IRX{n}$, let $\nFailX{\G}$ be the number
    of pairs $i< j$, such that there is no straight path between $i$
    and $j$ in $\G$. We refer to $\nFailX{\G}$ as the
    \emphi{deficiency} of $\G$.  Given a distribution $\DistrC$ over
    graphs, we use the shorthand
    $\nFailX{\DistrC} = \ExChar_{\G \sim \DistrC}\mleft[ \nFailX{\G}
    \mright]$

    For a parameter $\prS \in (0,1)$ and a number $n$, let
    $\lbY{n}{\prS} = \nFailX{\DistrY{n}{\prS}} $ be the \emphi{optimal
       deficiency}.  For a parameter $k$, a pair $i < j$ is a
    \emphi{$k$-hop failure} if there is no straight path from $i$ to
    $j$ with at most $\kk$ edges (i.e., $\kk$ hops).  Let
    $\lbZ{n}{\prS}{\kk}$ be the expected number of pairs $i<j$ that
    are $k$-hop failures for a graph drawn from $\DistrY{n}{\prS}$.
    The quantity $\lbZ{n}{\prS}{\kk}$ is the \emphi{optimal $k$-hop
       deficiency}.

\end{definition}

The optimal deficiency $\lbY{n}{\prS}$ is a lower bound on the
(expected) number of pairs with no straight path in a graph drawn from
$\DistrY{\G}{\prS}$, where $\G$ is an arbitrary graph over
$\IRX{n}$. The task at hand is to construct a graph $\G$, as sparse as
possible, such that $\nFailX{\DistrY{\G}{\prS}}$ is close to
$\lbY{n}{\prS}$.

\section{On the optimal deficiency \TPDF{$\lbY{n}{\prS}$}{l(n,p )}}
\seclab{opt:deficiency}

Consider the clique graph $K_n$ over $\IRX{n}$, where the weight of an
edge $ij$ is $|i-j|$. Here we investigate the expected number of pairs
(i.e., $\lbY{n}{\prS}$) that do not have a straight path in a graph drawn
from $\DistrY{K_n}{\prS}$.

\subsection{A rough upper bound}
\seclab{u:b:easy}

\begin{lemma}
    \lemlab{fail:edge}%
    For two indices $i < j$, with $\Delta = j - i$, let $\pdX{\Delta}$
    be the probability that there is \emph{no} $2$-hop straight path
    between $i$ and $j$ in $\G \sim \DistrY{K_n}{\prS}$.  We have
    \begin{equation*}
        (1-\prS)^{\Delta}%
        \leq%
        \pdX{\Delta}%
        =%
        (1-\prS) (1-\prS^2)^{\Delta -1}.
    \end{equation*}
\end{lemma}
\begin{proof}
    Let $q = 1-\prS$.  Let $\Event$ be the event that all the $\Delta$
    outgoing edges from $i$ to $i+1, \ldots, j$ are deleted. Let
    $\Event'$ be the symmetric event that all the $\Delta$ incoming
    edges into $j$ are deleted.  We have that $i$ is disconnected from
    $j$ if $\Event$, or even $\Event \cup \Event'$ happens. Observe
    that $\Prob{\Event} = \Prob{\Event'} = (1-\prS)^\Delta$, and
    \begin{align*}
      \pdX{\Delta}
      &\geq
        \Prob{ \Event \cup \Event' \bigr.}
        =%
        \Prob{ \Event} +
        \Prob{ \Event'} - \Prob{\Event \cap \Event'}
        =
        2(1-\prS)^{\Delta}  - (1-\prS)^{2\Delta-1}
      \\&
      =
      (1-\prS)^{\Delta} \pth{2  - (1-\prS)^{\Delta-2} }
      \geq
      (1-\prS)^{\Delta}.
    \end{align*}

    As for the exact probability, let
    $\Pi = \Set{ i t j }{ i < t < j }$ be the collection of
    $\Delta - 1$ $2$-hop straight path between $i$ and $j$ in
    $K_n$. These paths are edge disjoint, and the probability of each
    one of them to fail to be realized in $\G$ is exactly
    $1-\prS^2$. Thus, if there is no straight path from $i$ to $j$ in
    $\G$, then the edge $ij$ must be deleted, and so are all the paths
    of $\Pi$. This readily implies that
    \begin{math}
        \pdX{\Delta}%
        =%
        (1-\prS) (1-\prS^2)^{\Delta - 1},
    \end{math}
    as all these paths are disjoint.
\end{proof}

\begin{lemma}
    \lemlab{l:b:rough}%
    We have
    $\lbY{n}{\prS} = \nFailX{\DistrY{n}{\prS}} \leq \lbZ{n}{\prS}{2} \leq
    n/\prS^2$, see \defref{deficiency}.
\end{lemma}
\begin{proof}
    By \lemref{fail:edge}, we have that that the expected number of
    indices $j$, such that there is no $2$-hop path from a fixed
    $i <j$ to $j$ is
    \begin{math}
        \leq%
        \sum_{\Delta=1}^n \pdX{\Delta}%
        \leq%
        \sum_{\Delta=1}^n (1-\prS) (1-\prS^2)^{\Delta - 1}%
        \leq%
        \sum_{\Delta=0}^\infty (1-\prS^2)^{\Delta} \leq
        \frac{1}{1-(1-\prS^2)} = \frac{1}{\prS^2}.
    \end{math}
\end{proof}

We prove below a generalization of \lemref{l:b:k} for the optimal
$\kk$-hop deficiency, see \lemref{l:b:k} for details, for any
$\kk > 1$.

\subsection{A tighter upper bound on the deficiency
   \TPDF{$\lbY{n}{\prS}$}{l(n, pr )}}
\seclab{u:b:optimal}

In the following, $n$ and $\prS$ are parameters, and $\G$ is a graph
sampled from $\DistrY{n}{\prS}$.

The above analysis suggests that for two vertices to be connected (by
a straight path), with probability $\geq 1-\prS^{O(1)}$, in $\G$,
requires that their distance $\Delta$ has to be at least
$c' \prS^{-2} \log \prS^{-1}$, for some constant $c'$.  The lower
bound of \lemref{fail:edge}, on the other hand, implies that $\Delta$
must be at least $c' \prS^{-1} \log \prS^{-1}$.  It turns out that the
truth is closer to the lower bound, but proving it requires some work.

To this end, let
\label{nonsense}
\begin{align}
  \sB = \frac{c}{\prS}
  \Eqlab{c:value}%
\end{align}
(for simplicity, assume $\sB$ is an integer), where $c > 0$ is a
sufficiently large integer constant. We divide the vertices $\IRX{n}$
into $n / \sB$ blocks\footnote{Again, for simplicity of exposition, we
   assume $\sB$ divides $n$.}, each of size $\sB$, where the $i$\th
\emphi{block} is $B_i = \IRY{ (i-1) \sB +1}{ i \sB }$.

A vertex $i$ is \emphi{reachable} if there is a straight path from $1$
to $i$ in $\G$.  Let $\Rch = \RchX{\G}$ be the set of all reachable
vertices of $\G$. For a block $B$, let
$\goodX{B} = \cardin{ B \cap \Rch}$.

The following lemma testifies that (with good probability) the number
of reachable vertices grows exponentially between blocks, until it
reaches a constant fraction of the size of a block, and then it
remains stable.

The following lemma is implied by a standard application of Chernoff's
inequality.

\begin{lemma}\RefProofInAppendix{success}
    \lemlab{success}%
    Consider two blocks $B$ and $B'$, where $B$ appears before $B'$,
    and $\gC = \goodX{B} > 0$. We have:
    \begin{compactenumA}
        \smallskip%
        \item
        $\Prob{\bigl. \goodX{B'} \geq \min( 2\gC, \sB/3 ) } \geq 1/2$.

        \smallskip%
        \item If $\gC \geq \sB/3$, then
        $\Prob{\bigl. \goodX{B'} \geq (2/3)\sB } \geq 1/2$.
    \end{compactenumA}
\end{lemma}
\begin{proof:e}{\Xlemref{success}}{success}
    (A) Let $S = B \cap \Rch$, and observe that $\gC = \cardin{S}$.
    Consider the set of (directed) edges $E' =S \times B'$, where
    $B' = \IRY{t}{t + \sB -1 }$. For $i=1,\ldots, \sB$, let $X_i=1$
    $\iff$ there is an incoming edge of $E'$ into the vertex
    $t + i -1$ in the graph $\G$ (this implies that
    $t+i-1 \in B' \cap \Rch$).  We have that
    \begin{math}
        \gamma%
        =%
        \Prob{X_i = 1}%
        =%
        1- (1-\prS)^\gC.
    \end{math}
    Observe that $X_1, \ldots, X_\sB$ are independent, and their sum
    $Y = \sum_i X_i$ is concentrated, implied by Chernoff's
    inequality, as we show next (i.e.,, the remainder of the proof is
    by now standard tedium, and the reader might want to skip it).

    If $\gC \geq 1/\prS$, then
    \begin{math}
        \gamma%
        =%
        1- (1-\prS)^\gC
        \geq
        1 - \exp\pth{ -p \gC }
        \geq
        1 - \frac{1}{e}
        \geq %
        \frac{1}{2}.
    \end{math}
    Otherwise, if $\gC < 1/\prS$, then
    $(1-\prS)^{\gC-1} \geq (1-1/\gC)^{\gC-1} \geq 1/e$, and
    \begin{equation*}
        \gamma
        =%
        1- (1-\prS)^\gC%
        =%
        \prS \pth{ 1 + (1-\prS) + \cdots + (1-\prS)^{\gC-1} }%
        \geq%
        \prS \frac{u}{ e},
    \end{equation*}
    If $\gC \geq 1/\prS$, then
    \begin{math}
        \mu= \Ex{Y}%
        =%
        \sB \gamma%
        \geq%
        \sB /2,
    \end{math}
    and \thmrefY{chernoff}{Chernoff's inequality} implies that
    \begin{equation*}
        \Prob{Y \leq \sB/3}%
        \leq
        \Prob{ Y \leq (1-1/3) \mu}
        \leq
        \exp\bigl( -(1/3)^2\mu/2 \bigr)
        \leq%
        \exp \Bigl( -\frac{\sB}{36}\Bigr)
        =%
        \exp\Bigl( -\frac{c}{36p}\Bigr)
        \ll
        \frac{1}{2},
    \end{equation*}
    for $c > 36$, as $\sB = c/\prS$.  If $\gC < 1/\prS$, then
    $\mu = M \gamma \geq (c/\prS) \prS (\gC/e) = (c/e) \gC > 9 \gC$,
    for $c > 36$. As such, by \thmrefY{chernoff}{Chernoff's
       inequality}, we have
    \begin{equation*}
        \Prob{Y \leq 2u}%
        \leq
        \Prob{Y \leq (1-7/9) \mu }
        \leq%
        \exp\pth{ -(7/9)^2 \mu /2}
        \leq%
        \exp\pth{ -  \mu /4}
        \leq%
        \exp\pth{ -  2\gC}
        \leq
        \frac{1}{2},
    \end{equation*}
    since $\gC \geq 1$. As $\goodX{B'} \geq Y$, this completes the
    proof of this part.

    \bigskip%
    \noindent%
    (B) If $\gC \geq M/3 = c / (3\prS)$, then
    \begin{math}
        \gamma%
        \geq
        1 - \exp\pth{ -p \gC }
        =%
        1 - \exp\pth{ -c/3 }
        \geq
        1 - \exp\pth{ -12 }
        \geq %
        0.99999.
    \end{math}
    Thus, $\mu \geq 0.9999 \sB$, and by \thmrefY{chernoff}{Chernoff's
       inequality}, we have
    \begin{equation*}
        \Prob{\smash{Y \leq \frac{2}{3} \sB}\Bigr.}%
        \leq
        \Prob{\smash{ Y \leq \Bigl(1-\frac{1}{4}\Bigr) \mu } \Bigr.}
        \leq
        \exp\Bigl( -\frac{1}{4^2} \cdot \frac{\mu}{2} \Bigr)
        \leq%
        \exp \Bigl( -\frac{\sB}{33}\Bigr)
        =%
        \exp\Bigl( -\frac{c}{33p}\Bigr)
        \ll
        \frac{1}{2},
        \SoCG{\qedhere}
    \end{equation*}
\end{proof:e}

For $i \in \IRX{n/\sB}$, let $\gX{i} = \goodX{B_i}$, and let
$\prevX{i} = \arg\max_{j < i} \gX{j}$ be the index of the block before
$B_i$ with maximum reachable vertices.

\begin{definition}
    A block $B_i$ is \emphi{successful} if one of the following holds:
    \begin{compactenumI}
        \item $i =1$,
        \item $\gX{i} \geq \sB/3$, or
        \item for $j = \prevX{i}$, we have $\gX{i} \geq 2 \gX{j}$.
    \end{compactenumI}
\end{definition}

\begin{remark}
    \remlab{fix}%
    By \lemref{success}, a block $B_i$, for $i > 1$, has probability
    at least half to be successful. Intuitively, whether or not two
    blocks are successful is almost an independent event, and in
    particular the total number of successful blocks is strongly
    concentrated. Here are two ways to prove this formally:

    \begin{compactenumA}
        \item \textbf{Azuma's inequality}: Let $X_i$ be an indicator
        variable that is $1$ if $B_i$ is successful. Reinterpreting
        \lemref{success}, we have that
        $\rho_i = \ExCond{X_i}{X_1, \ldots, X_{i-1}} =
        \ProbCond{X_i}{X_1, \ldots, X_{i-1}} \geq 1/2$. For a fixed
        $t$, let $Z = \sum_{i=1}^t X_i$. For $i=0,\ldots, t$, let
        $Y_i= \ExCond{Z}{X_1, \ldots, X_i}$. It is well known that the
        sequence of $Y_i$s form a martingale. Furthermore, we have
        that $\cardin{Y_i - Y_{i-1}} \leq 1$ (indeed, when the value
        of $X_i$ is exposed, the expectation for the total sum of the
        $X_i$s can change by at most $1$). Thus, for any
        $\Delta \geq 0$, by Azuma's inequality, we have
        \begin{equation*}
            \Prob{\bigl. Z < \Ex{Z} - \Delta}
            =%
            \Prob{\bigl. \Ex{Z} - Z > \Delta}
            =%
            \Prob{Y_0 - Y_t > \Delta}%
            \leq
            \exp \pth{ -\Delta^2 /(2t)}.
        \end{equation*}

        \medskip%
        \item \textbf{Coupling with an independent sequence.} A more
        elegant argument (but somewhat longer) is via coupling. We
        determine the values of the $X_i$s one by one. Assume we are
        in the $i$\th step of this process, where
        $X_1, \ldots, X_{i-1}$ were already determined, and $X_i$ is
        about to be fixed (i.e., the relevant edges in the graph
        connecting $B_i$ to earlier blocks, and the edges inside $B_i$
        are being exposed). Let
        $p_i =\ProbCond{X_i}{X_1, \ldots, X_{i-1}}$. By
        \lemref{success}, we have $p_i \geq 1/2$. Thus, if $X_i$ is
        set to be $0$, then we set $G_i=0$. Otherwise, if $X_i=1$, we
        set $G_i=1$ (independently) with probability $1/(2p_i)$, and
        otherwise set $G_i=0$.  Observe that
        $\Prob{G_i=1} = \ProbCond{X_i=1}{X_1, \ldots, X_{i-1}}
        \tfrac{1}{2p_i} = p_i \cdot \tfrac{1}{2p_i} = \tfrac{1}{2}$.
        Clearly, the sequence $G_1, \ldots, G_t$, has the distribution
        of $t$ independent unbiased coin tosses, and this coupling
        argument implies that
        $\sum_{i=1}^t X_i \geq \sum_{i=1}^t G_i $.
    \end{compactenumA}
\end{remark}

\begin{lemma}
    \lemlab{disconnect}%
    Let $c_5 > 1 $ be some integer constant.  Consider two vertices
    $i$ and $j$, such that $j > i + t c_5 \sB$, and
    $t > 1 + \ceil{ \log_2 \sB} = \Theta( \log\tfrac{1}{\prS})$ is an
    integer. The probability that there is no straight path from $i$
    to $j$ in $\G$ is at most $\exp(-t)$, for $c_5$ sufficiently
    large.
\end{lemma}
\begin{proof}
    We might as well assume that $i=1$. There are at least
    \begin{equation*}
        \nu = c_5 t-1 \geq (c_5/2)\ceil{\log_2 (1/\prS)}
    \end{equation*}
    blocks between $1$ and $j$ in $\G$: $B_2, B_3, \ldots,
    B_{\nu+1}$. Let $X_i=1$ if the block $B_{i+1}$ is successful, and
    $X_i=0$ otherwise.  As suggested by \remref{fix} (B), we work with
    the coupled independent random variables $G_1, \ldots, G_\nu$
    instead of $X_1, \ldots, X_n$, as $\sum_i X_i \geq \sum_i
    G_i$. Here, the $G_i$s are independent and $\Prob{G_i=1} =1/2$.
    For $Z = \sum_i X_i$ and $Y = \sum_i G_i$, we have
    $\Ex{Z} \geq \Ex{Y} = \nu/2$. By Chernoff's inequality, we have
    \begin{align*}
      \beta_1%
      &=%
        \Prob{ Z \leq 2t}
        \leq
        \Prob{ Y \leq 2t}
        \leq%
        \Prob{ Y \leq \frac{\nu}{2} \cdot \frac{4t}{\nu}
        }
        =%
        \Prob{ Y \leq \Bigl(1-\frac{\nu - 4t}{\nu}\Bigr)\frac{\nu}{2}}
      \\&%
      \leq%
      \exp\Bigl( - \Bigl(\frac{\nu - 4t}{\nu}\Bigr) ^2 \frac{\nu}{4} \Bigr)
      \leq%
      \exp\Bigl( - \frac{\nu}{8} \Bigr)
      \leq
      \exp\pth{ -t-1},
    \end{align*}
    by picking $c_5 \geq 16$.  Thus, with good probability, $Y > 2t$
    and there are ``many'' successful blocks.  Let
    $i_1, i_2, \ldots, i_Y$ be the indices of these successful
    blocks. In particular, for $j=1,2,\ldots$,
    $\goodX{i_{j+1}}\geq 2\goodX{i_{j}}$, till
    $\goodX{i_j} \geq \sB/3$. Thus, for all
    $j \geq t \geq 1+ \log_2 \sB$, we have that $\gX{i_j} \geq
    \sB/3$. Namely, there are at least $t$ ``heavy'' blocks with at
    least $\sB/3$ reachable vertices, in each one of them, all of them
    appearing before $j$.

    The probability that all the edges, from one of these reachable
    vertices of a heavy block to $j$, fail to be selected in $\G$
    (which was sampled from $\DistrY{n}{\prS}$), is at most
    $(1-\prS)^{\sB/3}$. The probability that all the heavy blocks
    fail in this way is thus at most
    \begin{equation*}
        \beta_2
        =
        \pth{ (1-\prS)^{\sB/3} }^{t}
        \leq
        \exp\pth{ -\frac{\prS \sB}{3} \cdot t}
        =%
        \exp\pth{ - \frac{c t}{3}}
        \leq%
        \exp(-t-1).
    \end{equation*}
    Thus, the probability that $j$ is not reachable is at most
    $\beta_1 + \beta_2 \leq e^{-t}$.
\end{proof}

\begin{lemma}
    We have
    $\lbY{n}{\prS} = \nFailX{\DistrY{n}{\prS}} = O\bigl( (n/\prS) \log
    (1/\prS) \bigr)$, see \defref{deficiency}.
\end{lemma}
\begin{proof}
    Fix a vertex $i$, and let $c_5$ be the constant from
    \lemref{disconnect}.  Let
    $\nu = 1 + 2\ceil{\ln (c_5 \sB)} = \Theta( \log \tfrac{1}{\prS} )$
    and $T = \ceil{ c_5 \sB} = \ceil{ c_5 c /\prS}$, we
    (conservatively) count all the vertices in the range
    $\IRY{i+1}{ i + \nu T}$ as not being reachable from $i$ in $\G$.
    As for bounding the expected number of unreachable vertices
    further away, we apply \lemref{disconnect} to the super blocks of
    size $T$, where $i > \nu$, which implies the upper bound
    \begin{equation*}
        \sum_{t =  \nu+1}^{\infty}  \exp( -t ) T %
        =
        c_5 \sB\sum_{t = \nu+1}^{\infty}  \exp( -t) %
        \leq
        \frac{c_5 \sB}{(c_5 \sB)^2}
        \leq%
        1.
    \end{equation*}
    Namely, in expectation, $i$ can not reach (via a straight path) at
    most $T + 1 = O\bigl( \tfrac{1}{\prS} \log \tfrac{1}{\prS} \bigr)$
    vertices. Adding this bound over all $i$ implies the claim.
\end{proof}

\subsection{The lower bound}

\begin{lemma}
    \lemlab{l:b:clique}%
    We have
    $\lbY{n}{\prS} = \nFailX{\DistrY{n}{\prS}} = \Omega( (n/\prS) \log
    (1/\prS) )$, see \defref{deficiency}.
\end{lemma}
\begin{proof}
    Assume $\prS =1/u$ where $u$ is some integer $\geq 8 \geq e^2$.
    The expected number of vertices of $B_i$ that are reachable by a
    straight-path from $1$ in the range $\IRX{t}$, for
    $t=\tfrac{1}{\prS} \ln \tfrac{1}{\prS} \geq \tfrac{2}{\prS}$ can
    be bounded by the expected number of straight paths from $1$ to
    any number in $\IRX{t}$. The key observation is that there are
    $\binom{t-1}{k}$ such paths with $k$ hops, and each such path has
    probability exactly $\prS^k$ to be in the random graph.  Indeed,
    such a $k$-hop path starting at $1$, involves choosing $k$ indices
    $1 < i_1 < i_2 < \ldots < i_k$ all in the range
    $\IRY{2}{t} \subseteq \IRX{t}$, where $i_k$ is the destination of
    the path.

    Thus, the expected number of reachable vertices in $\IRX{t}$ from
    $1$ is bounded by
    \begin{equation*}
        \sum_{k=0}^{t} \binom{t}{k} \prS^k
        =%
        (1+\prS)^{t}
        \leq
        \exp( \prS t)
        \leq
        \exp\Bigl( \ln \frac{1}{\prS} \Bigr)
        =
        \frac{1}{\prS}
        \leq
        \frac{t}{2}.
    \end{equation*}
    Namely, at least half the vertices in the range $\IRX{t}$ are not
    reachable from $1$ in expectation. By linearity of expectations,
    this implies that (in expectation), at least
    $\Omega\bigl( (n/2) (t/2) \bigr) = \Omega( \tfrac{n}{\prS} \ln
    \tfrac{n}{\prS} )$ pairs in $\IRX{n}$ are not reachable to each
    other via a straight path.
\end{proof}

\subsection{The result}

Putting the above together, we get the following result.

\begin{theorem}
    \thmlab{clique}%
    Let $n>1$ and $\prS \in (0,1)$ be parameters.  We have
    $\lbY{n}{\prS} = \nFailX{\DistrY{n}{\prS}} = \Theta( (n/\prS) \log
    (1/\prS) )$, see \defref{deficiency}.
\end{theorem}

\section{Constructing exact spanner in \TPDF{$1$}{1}-dim}
\seclab{1:dim}%

Our purpose is to build an exact spanner on $\IRX{n}$, such that its
number of edges is near linear, and its deficiency (i.e., the expected
number of failed pairs) is close to the optimal deficiency (i.e.,
\thmref{clique}). It is useful to start with a lower bound.

\subsection{A lower bound on the number of failed pairs in a graph}
\seclab{l:b:graph}

\begin{lemma}
    Let $\G$ be a graph with $n$ vertices, and
    $m \leq ({n}/{8})\log_{1/(1-\prS)}n = \Theta( \tfrac{n}{\prS} \log
    n)$ edges.  Let $\H$ be a graph randomly sampled from
    $\DistrY{\G}{\prS}$.  Then $\Ex{\nFailX{\H}}\geq n^{3/2}/8$.
\end{lemma}
\begin{proof}
    The average degree of $\G$ is
    \begin{math}
        d =%
        \frac{2m}{n} \leq \frac{1}{4}\log_{1/(1-\prS)} n.
    \end{math}
    Let $U$ be the set of vertices of $\VX{\G}$ that are of degree
    $\leq 2d$. By Markov's inequality, we have $|U| \geq n/2$. The
    probability that all the edges attached to a vertex $u \in U$
    fail, is at least
    \begin{equation*}
        (1-\prS)^{2d}
        =%
        (1-\prS)^{(\log_{1/(1-\prS)} n)/2}
        =
        \frac{1}{\sqrt{n}}.
    \end{equation*}
    Thus, the expected number of isolated vertices in $U$ is at least
    $|U|/\sqrt{n} \geq \sqrt{n}/2$. Each such isolated vertex induces
    $n-1$ failed pairs. We conclude that the expected number of failed
    pairs is at least $(\sqrt{n}/2) (n-1)/2 \geq n^{3/2}/8$.
\end{proof}

Thus, for $\prS$ bounded away from $0$, any graph with
$\Ex{\nFailX{\G}} = O(n)$ (i.e., only a linear number of failed
pairs), must have $\Omega(n \log n)$ edges.

\subsection{Exact spanner in \TPDF{$1$}{1}-dim with long paths}
\seclab{long:hop:1:dim}

\myparagraph{The construction.} %
The input is a number $n$, and a parameter $\prS \in (0,1)$. One can
safely assume that $n \geq 1/\prS$. Let $\G$ be the graph over
$\IRX{n}$, where we connect any two vertices if they are in distance
at most $\sL = \ceil{ ({\cODS }/{\prS})\log n }$ from each other,
where $\cODS$ is a sufficiently large constant.

\begin{lemma}
    \lemlab{long:paths}%
    Let $\G$ be the above constructed graph with
    $\cODS \tfrac{n}{\prS}\log n $ edges.  Then,
    $\nFailX{\bigl.\DistrY{\G}{\prS}} \leq \lbY{n}{\prS} +1$, for
    $\cODS$ a sufficiently large constant.
\end{lemma}
\begin{proof}
    For any integer $i$, the induced subgraphs of $\G$ and $K_n$
    restricted to $\IRY{i}{i+\sL}$ are identical (i.e., $\G$ is a
    ``local'' clique for any such consecutive set of vertices). In
    particular, the analysis of \lemref{disconnect} implies that the
    probability of any two vertices $i < j$, such that
    $\sL/4 < |i-j| < \sL$, to be a failed pair (i.e., there is no
    straight path from $i$ to $j$) is at most (say) $1/n^{8}$, by
    making $\cODS$ sufficiently large.

    So, let $\H$ be a graph sampled from $\DistrY{\G}{\prS}$.  A pair
    $i<j$ is \emphw{short}, if $|i-j| \leq \sL$, and otherwise it is
    \emphw{long}. Since $\G$ and $K_n$ look the ``same'' for short
    pairs, it follows that the expected number of short failed pairs
    in $\H$ and in a random graph of $\DistrY{n}{\prS}$ are the same.

    As for a long pair $i <j$, there are at most $\binom{n}{2}$ such
    pairs. Each such pair can be connected by a path made out of
    medium length edges. Specifically, we choose indices
    $i = i_1 < i_2 < \cdots < i_k=j$, such that for all $t$, we have
    $\sL > i_{t +1} - i_t > \sL/4$. All these medium-length pairs are
    good, with probability $\geq 1 - \binom{n}{2}/n^8 \leq 1-1/n^6$.
    This implies that all the long pairs are reachable via these
    medium-length ``paths''.

    We conclude that
    \begin{math}
        \nFailX{\DistrY{\G}{\prS}}%
        \leq%
        \lbY{n}{\prS} + \binom{n}{2} /n^6 \leq%
        \lbY{n}{\prS} + 1.
    \end{math}
\end{proof}

\section{Constructing small-hop exact spanner in \TPDF{$1$}{1}-dim}
\seclab{1:dim:few:hops}

The above construction has a large diameter as far as number of
hops/edges. We want a better construction that has a small-hop
diameter.

\subsection{Handling the short pairs}
\seclab{short:pairs}

\begin{definition}[Parameters.]
    \deflab{params}%
    Let $\cKHop$ be a sufficiently large constant (its value would be
    determined shortly). The input parameters are $n$ and $\prS$, and
    let
    \begin{equation}
        \nbl = \frac{1}{\prS^{1/3}},%
        \qquad
        \sB
        =
        \sBVal%
        =%
        \ceil{\frac{\cKHop}{\prS^{4/3}} \ln n}
        \qquad\text{and}\qquad%
        \sL
        =
        6 \sB.%
        \Eqlab{s:b:val}
    \end{equation}
\end{definition}

Consider the graph $\G$ over $\IRX{n}$ where we connect two vertices
$i <j$ if
\begin{math}
    \cardin{i-j} \leq \sL.
\end{math}

\begin{lemma}
    For $\G$ the graph constructed above, let
    $\H \sim \DistrY{\G}{\prS}$.  The expected number of pairs
    $i<j < i + \sL$, such that this pair has no $2$-hop path in $\H$,
    is bounded by $n/\prS^2$.
\end{lemma}
\begin{proof}
    \lemref{fail:edge} quantify the probability of having a $2$-hop
    path between two vertices.  Now, the result is readily implied by
    the analysis of \lemref{l:b:rough}.
\end{proof}

\subsection{Handling the long pairs}

We need a construction of a $2$-hop exact spanner on $\IRX{n}$.  There
are some beautiful constructions known of bounded hop spanners \cite{
   cfl-ufica-85, c-cftcp-87, as-opaol-24}. For example, Chazelle
\cite{c-cftcp-87} shows how to construct a graph with $O(n)$ edges,
and $O(\alpha(n))$-hop diameter, where $\alpha(\cdot)$ is the inverse
Ackermann function.  Happily\footnote{Or maybe sadly?}  the simple
$2$-hop construction (that we describe next) is sufficient for our
purposes.

\begin{lemma}
    \lemlab{2:hop}%
    One can construct a $2$-hop exact spanner on $\IRX{n}$ with $O(n
    \log n)$ edges.
\end{lemma}
\begin{proof}
    Take the median $m = \floor{n/2}$, and connect it to all the
    vertices in $\IRX{n}$ except itself, by adding $n-1$ edges to the
    graph. Now, continue the construction recursively on
    $\IRY{1}{m-1}$ and $\IRY{m+1}{n}$. Let $\G$ be the resulting
    graph. The recursion on the number of edges of $\G$ is
    $E(n) = n-1 + 2 T(\floor{n/2}) = O( n \log n)$. As for the $2$-hop
    property, consider any $i < j$, and let $I=\IRY{\alpha}{\beta}$ be
    the lowest recursive subproblem still having $i$ and $j$ in the
    subproblem. Let $t$ be the median of $I$. If $t=i$ or $t=j$ then
    $ij \in \EGX{\G}$. Otherwise, $it, tj \in \EGX{\G}$, and
    $i <t < j$.
\end{proof}

We break $\IRX{n}$ into consecutive blocks, where each block has size
$\sB$. Thus, the $i$\th \emphi{block} is
$B_i = \IRY{(i-1) \sB +1}{i\sB}$. Let
$\BSet = \set{ B_1, \ldots, B_{n/\sB}}$ be the resulting set of
blocks.  We construct the graph of \lemref{2:hop} where $\BSet$ is the
set of vertices. We then take this graph, and every edge $B_i B_j$ is
replaced by the bipartite clique $B_i \otimes B_j$, see
\defref{b:clique}. Let $\H$ be the graph resulting form adding all
these bicliques to the graph $\G$ constructed in
\secref{short:pairs}. We claim that $\H$ is the desired graph.

\begin{claim}
    \clmlab{spanner}%
    For $n, \prS$ parameters, the graph $\H$ constructed above has
    $O(\tfrac{n}{\prS^{3/2}} \log^2 n)$ edges, Furthermore,
    $\nFailX{\DistrY{\H}{\prS}} \leq \lbY{n}{\prS} +1$, and for a random
    graph $\K \in \DistrY{\H}{\prS}$, and any $i < j$, such that
    $j > i + \sL$, we have that there is a straight path from $i$ to
    $j$ in $\K$ with at most $4$-hops. This holds with high
    probability for all such pairs.
\end{claim}
\begin{proof}
    Let $B_s$ and $B_t$ be the two blocks containing $i$ and $j$,
    respectively.  By construction, there is a middle block $B_m$,
    such that $B_{s+1} \otimes B_m$ and $B_m \otimes B_{t-1}$ are
    present in the graph.  One can show, that with high probability,
    there exists $i_1 = i$, $i_2 \in B_{s+1}$, $i_3 \in B_{m}$,
    $i_4 \in B_{t-1}$ and $i_5 =j$, such that $i_1 i_2 i_3i_4i_5$
    exists in $\K$ with high probability. We omit proving this here,
    as it follows from the analysis below.

    As for size, we have that the graph $\H$ has
    \begin{math}
        O\pth{n \sL + \pth{\frac{n}{\sB} \log \frac{n}{\sB}} \cdot
           \sB^2 } =%
        O\pth{\frac{n}{\prS^{4/3}} \log^2 n}
    \end{math}
    edges.~
\end{proof}

\subsection{Reducing the size: Constructing sparse bipartite
   connectors}

To reduce the number of edges in the above graph $\H$, the idea is to
replace the bicliques by a bipartite expander. This expander is
generated by picking each edge of the biclique randomly with
probability $\pr$.  A key tool in our analysis is understanding how
connectivity behaves between two disjoint sets, when edges are chosen
randomly.

\begin{lemma}
    \lemlab{silly:2}%
    Let $B,C \subseteq \IRX{n}$ be two disjoint sets, and consider the
    bipartite graph $\G = B \otimes C$. Let $\pr$ be some
    probability. Let $\H \sim \DistrY{\G}{\pr}$. Let $Y$ be the number
    of vertices in $C$ that have an adjacent edge in $\H$. Then, we
    have that
    \begin{compactenumI}
        \smallskip%
        \item
        $\mu = \Ex{Y} \geq M = \cardin{C}(1- \exp \pth{ - \pr
           \cardin{B}})$.  \smallskip%
        \smallskip
        \item
        $\Prob{\bigl.Y \leq (3/4)\Ex{Y}} \leq \exp( - \mu / 4 ) $.
        \smallskip
        \item $\Prob{\bigl.Y \leq \Ex{Y}/2} \leq \exp( - \mu / 8 ) $.
    \end{compactenumI}
\end{lemma}
\begin{proof}
    Let $\beta = \cardin{B}$.  Let $X_i = 1$ $\iff$ there is an edge
    in $\H$ that enters the $i$\th vertex of $C$ (otherwise
    $X_i=0$). As $1-x \leq \exp(-x)$, we have
    \begin{equation*}
        \Prob{X_i=1} = 1 - (1-\pr)^\beta%
        \geq%
        1 - \exp( -\pr \beta ).
    \end{equation*}
    Namely, for $Y= \sum_{i=1}^{\cardin{C}} X_i$, we have
    \begin{math}
        \mu =%
        \Ex{Y} \geq%
        \cardin{C} (1- \exp\pth{ - \pr \beta }).
    \end{math}

    Finally, by \thmrefY{chernoff}{Chernoff's inequality}, we have
    \begin{math}
        \Prob{ Y \leq \frac{3}{4} \Ex{Y}} \leq \exp\pth{ - \mu (9 /32)}.
    \end{math}
\end{proof}

\begin{definition}
    \deflab{expander}%
    Let $X,Y \in \BSet$ be two distinct blocks, and let
    \begin{equation*}
        \tau = \frac{\cKHop^2 \nbl}{\prS \sB}
    \end{equation*}
    (see \defref{params}).  A random graph
    $\GGen(X,Y) \sim \DistrY{ X\otimes Y}{\tau}$ is a \emphi{bipartite
       connector} between $X$ and $Y$, see \defref{b:clique}.

    For a graph $\G$, and a set $S \subseteq \VX{\G} $, let
    \begin{math}
        \Gamma_\G( S ) =%
        \Set{ y \in \VX{\G}}{x \in S, xy \in \EGX{ \G}}
    \end{math}
    be the set of \emphw{neighbors} of $S$ in $\G$.
\end{definition}

\begin{fact}
    \factlab{exp}%
    For $x \in (0,1)$, we have
    \begin{math}
        \exp( -\frac{x}{1-x} ) \leq%
        1-x \leq \exp(-x) \leq 1-x/2.
    \end{math}
\end{fact}

\begin{remark}[Intuition]
    There are too many parameters flying around, so lets try to
    understand what is going on. Consider two blocks $L,R$ that are
    ``far'' from each other, that we construct the random connector
    $\G = \GGen(L,R) \sim \DistrY{ L\otimes R}{\tau} $ between
    them. Simulating the edge deletion with probability $\prS$ on
    $\G$, is equivalent to sampling a graph
    $\H \sim \DistrY{ \G }{\prS}$. This is equivalent to directly
    sampling $\H \sim \DistrY{ L\otimes R}{\pr}$, where
    \begin{equation*}
        \pr =%
        \tau \prS%
        =%
        \prS \tfrac{\cKHop^2\nbl}{\prS \sB} = \tfrac{\cKHop^2\nbl}{
           \sB}%
        \qquad\implies\qquad
        (\cKHop-1) \tfrac{\prS}{ \ln n}
        \leq
        \pr \leq
        \cKHop \tfrac{\prS}{ \ln n},
    \end{equation*}
    since $\nbl = \frac{1}{\prS^{1/3}}$, $\sB=\sBValS$, and
    $\tau = \frac{\cKHop^2 \nbl}{\prS \sB}$. In particular, the degree
    of each vertex of $L \otimes R$ is $\sB$, and thus the expected
    degree of a vertex in $\H$ is
    \begin{equation}
        \mu_d
        =
        \sB \pr
        =
        \Theta\pth{ \tfrac{\cKHop \nbl}{\prS} \ln n \cdot
           \cKHop \tfrac{\prS}{ \ln n}
        }
        =%
        \Theta(\cKHop^2 \nbl )
        =
        \Theta( \cKHop^2 / \prS^{1/3} ).
    \end{equation}
    Namely, for a constant $\prS$, $\H$ (informally) is a constant
    degree bipartite random graph --- that is, a bipartite expander.
\end{remark}

The following (tedious) lemma, quantify the expansion of the residual
graph $\H$.

\begin{lemma}\RefProofInAppendix{expander}
    \lemlab{expander}%
    Let $\nbl$, $\sB$, $\sL$ and $\tau$ be the parameters as specified
    in \defref{params} and \defref{expander}.  Let $L,R \in \BSet$ be
    two distinct blocks (of size $\sB$), and consider a random
    bipartite connector $\G= \GGen(L,R)$. Let
    $\H \sim \DistrY{\G}{\prS}$.  We have the following properties:
    \begin{compactenumA}
        \smallskip%
        \item  \itemlab{jump:2} For a set $S \subseteq L$, with
        $\cardin{S} \geq \prS\sB/2$, we have
        $ \cardin{\Gamma_\H( S)} > {\prS}^{2/3}\sB/4 = \Omega( \nbl^2
        \ln n )$.

        \smallskip%
        \item \itemlab{jump:3} For a set $S \subseteq L$, with
        $\cardin{S} \geq \prS^{2/3}\sB$, we have
        $ \cardin{\Gamma_\H( S)} > \prS^{1/3}\sB/2 = \Omega( \nbl^3
        \ln n ) = \Omega( \tfrac{1}{\prS} \ln n )$.
    \end{compactenumA}
    \smallskip%
    Each of the above holds with probability $\geq 1-1/n^{O(1)}$.
\end{lemma}
\begin{proof:e}{\Xlemref{expander}}{expander}
    As a reminder of the parameters, we have
    $\nbl = \frac{1}{\prS^{1/3}}$, $\sB=\sBValS$, and
    $\tau = \frac{\cKHop^2 \nbl}{\prS \sB}$, and
    $\H \sim \DistrY{ L \otimes R}{\pr}$, where
    \begin{math}
        \pr \geq%
        (\cKHop-1) \tfrac{\prS}{ \ln n} .
    \end{math}
    Consider the subgraph of $\G$ involving only $S$ on the left side.
    Somewhat mysteriously, observe that
    \begin{equation*}
        \pr |S|
        \geq
        \frac{\prS \sB}{2}
        \cKHop \frac{\prS}{2 \ln n}
        \geq
        \cKHop^2 \nbl \ln n \cdot
        \frac{\prS}{4 \ln n}
        =
        \cKHop^2 \frac{\prS^{2/3}}{4 }.
    \end{equation*}
    Let $Z= \cardin{\Gamma_\H( S)}$. By \lemref{silly:2} and
    \factref{exp}, we have
    \begin{equation*}
        \Ex{\bigl.Z}%
        \geq%
        M = \sB(1- \exp \pth{ - \pr
           \cardin{S}})
        \geq
        \sB\bigl(1- \exp ( - (\cKHop^2 /4) \prS^{2/3}) \bigr)
        \geq
        \frac{\sB  \prS^{2/3}}{2}%
        \geq%
        \frac{\cKHop}{2} \ln n,
    \end{equation*}
    and
    \begin{math}
        \Prob{Z \leq M/2} \leq \Prob{Z \leq \prS^{2/3} \sB/4} \leq%
        \Prob{\bigl. Z \leq \Ex{Z}/2} \leq \exp( -M/8 ) \leq
        1/n^{\Theta(\cKHop )}.
    \end{math}

    \medskip

    (B) Observe that
    \begin{math}
        \pr |S| \geq%
        \tfrac{\cKHop^2\nbl}{ \sB} \cdot \frac{\prS^{2/3}}{2}\sB =%
        \cKHop^2 \prS^{1/3} \geq \prS^{1/3}.
    \end{math}
    The claim now follows by the argument as above.
\end{proof:e}

\subsection{The improved \TPDF{$4$}{4}-hop \dependable exact spanner
   construction}

We first construct the graph $\G_1$ over $\IRX{n}$ as described in
\secref{short:pairs}. Next, break $\IRX{n}$ into consecutive blocks,
where each block has size $\sB = \sBValS$. Thus, the $i$\th block is
$B_i = \IRY{(i-1) \sB +1}{i\sB}$, and let
$\BSet = \{ B_1, \ldots, B_{n/\sB} \}$ be the resulting set of blocks.
We construct the graph $\G_\BSet$ of \lemref{2:hop} on $\BSet$ (i.e.,
a vertex in this graph is a block of $\BSet$). For every edge
$B_i B_j \in \EGX{G_\BSet}$ is replaced by the bipartite random
connector of \defref{expander}, that is $\GGen(B_i, B_j)$, and in
particular, we add the edges of this graph, to the graph $\G_1$. Let
$\G$ be the resulting graph.

By using the random connector instead of a bipartite clique, we
reduced its size.

\begin{lemma}\RefProofInAppendix{n:edges}
    \lemlab{n:edges}%
    The graph $\G$ has (in expectation)
    $\displaystyle O\Bigl( \frac{n}{\prS^{4/3}} \log n\Bigr)$ edges.
\end{lemma}
\begin{proof:e}{\Xlemref{n:edges}}{n:edges}
    Since the initial graph $\G_1$ connects only vertices in distance
    $\leq \sL$ from each other, it has at most $O(n \sL)$ edges, where
    $\sL = \Theta( \tfrac{1}{\prS^{4/3}} \log n)$ by
    \defref{params}. Each vertex of the graph $\GGen(X,Y)$, from
    \defref{expander}, has in expectation degree
    \begin{math}
        \sB \cdot \tau =%
        \sB \frac{\cKHop^2 \nbl}{\prS \sB} =%
        \Theta\bigl( \frac{1}{\prS^{4/3}} \bigr).
    \end{math}
    Thus, overall in expectation, the graph $\GGen(X,Y)$ has
    $O( \sB/\prS^{4/3})$ edges.  The graph $\G_\BSet$ has
    $O\pth{ \frac{n}{\sB} \log \frac{n}{\sB}}$ edges.  Thus, the
    edge-replacement process for $G_\BSet$ adds at most
    \begin{equation*}
        O\pth{ \frac{n}{\sB} \log \frac{n}{\sB} \cdot
           \frac{\sB}{\prS^{4/3} } }%
        =%
        O\pth{ \frac{n}{\prS^{4/3} } \log n
        }%
    \end{equation*}
    edges to $\G_1$.
\end{proof:e}

\begin{lemma}%
    \lemlab{ex:number}%
    Let $\G$ be the above constructed graph with
    $O( \tfrac{n}{\prS^{4/3}}\log n )$ edges.  Then, for a random
    graph $\H \in \DistrY{\G}{\prS}$, we have,
    $\nFailX{\DistrY{\G}{\prS}} = \Ex{\nFailX{\H}} \leq \lbY{n}{\prS}
    +1$.  Furthermore, the expected number of pairs $i <j$ such that
    there is no straight $\leq 4$-hop path from $i$ to $j$ in $\H$, is
    $\leq n/\prS^2 + 1$.
\end{lemma}
\begin{proof}%
    The graph $\G$ contains the corresponding graph of
    \lemref{long:paths}, and the first part of the claim readily
    follows.

    So consider a long pair $i + \sL < j$. We claim that there is a
    $4$-hop straight path from $i$ to $j$.  In particular, there are
    at least (say) $8$ blocks between $i$ and $j$. Formally, let
    $B_\alpha, B_\beta \in \BSet$ be the two blocks containing $i$ and
    $j$, respectively.

    A key property is that, by construction, all the vertices of two
    adjacent blocks are connected (by a clique), by construction.
    Thus $i$, is connected by edge to all the vertices of
    $B_{\alpha+1}$. As such, the vertex $i$ is connected (in
    expectation) to
    \begin{equation*}
        \mu =
        \prS \cardin{B_{\alpha+1}}
        =%
        \prS \sB%
        \geq
        \cKHop \nbl \ln n
        =
        \cKHop \frac{ \ln n}{\prS^{1/3}}.
    \end{equation*}
    vertices of $B_{\alpha+1}$ by a direct edge, a as
    $\nbl = \frac{1}{\prS^{1/3}}$ and $\sB=\sBValS$.

    Thus, by Chernoff's inequality (with high probability), there are
    at least $\prS \sB/2$ vertices in $B_{\alpha+1}$ that are
    reachable from $i$ by a direct edge, and let $S_1$ be this set of
    vertices.  The $2$-hop graph $\G_\BSet$, has a middle block
    $B_\gamma$, such that we constructed the random connectors
    $\GGen(B_{\alpha+1}, B_\gamma)$ and
    $\GGen(B_{\gamma}, B_{\beta -1})$, with $\alpha < \gamma < \beta$.

    By \lemref{expander} \itemref{jump:2}, applied to
    $(B_{\alpha+1}\otimes B_\gamma) \cap \H$ and $S_1$, there is a set
    $S_2 \subseteq B_\gamma$ of size at least $\prS^{2/3} \sB/2$, that
    all its vertices are reachable from $i$ by a $2$-hop straight path
    in $\H$.

    By \lemref{expander} \itemref{jump:3}, applied to
    $(B_\gamma\otimes B_{\beta-1}) \cap \H$ and $S_2$, there is a set
    $S_3 \subseteq B_{\beta-1}$ of size at least
    $\prS^{1/3} \sB/2 = \Omega( \tfrac{1}{\prS}\ln n)$, such that all
    its vertices are reachable from $i$ by $3$-hop straight path in
    $\H$.  Each one of the above three groups have the desired bounds
    on their size with high probability in $n$.

    Finally, the probability that none of the vertices of
    $S_3$ have direct edge into $j$ is at most
    \begin{equation*}
        (1-\prS)^{|S_3|}
        \leq%
        \exp\Bigl(-\frac{\prS^{1/3} \sB}{2} \cdot \prS \Bigr)
        <%
        \exp\Bigl(-\frac{\prS^{4/3} }{2}
        \cdot
        \frac{\cKHop}{\prS^{4/3}} \ln n
        \Bigr)
        <%
        \frac{1}{n^{O(\cKHop )}},
    \end{equation*}
    see \Eqref{s:b:val}.

    We conclude that all long pairs have a $4$-hop paths between them
    with high probability.  The expected number of failed $\leq 4$-hop
    pairs that are short is bounded $n/\prS^2$, by \lemref{l:b:rough}, as
    $\G$ looks locally like a clique if $j < i + \sL$.
\end{proof}%

We summarize the result the above implies.

\begin{lemma}
    \lemlab{4:hops}%
    Let $n > 0$ be an integer, and $\prS \in (0,1)$ a parameter. The
    above constructed graph $\G$ over $\IRX{n}$ has
    $O( (n/\prS^{4/3}) \log n)$ edges. Furthermore, for a graph
    $\H \sim \DistrY{\G}{\prS}$, we have that $\H$ provides a
    $\leq 4$-hop path for all pairs except (in expectation)
    $\leq n/\prS^2 + 1$ pairs. For all pairs $i,j$, with
    $j > i + \Omega( \prS^{-4/3} \log n )$, such a $4$-hop straight
    path exists with high probability.
\end{lemma}

\subsection{Dependable \TPDF{$\kk$}{k}-hop \TPDF{$1$}{1}-dim spanner}

One can tradeoff the dependency on $\prS$ by allowing more hops for
the spanner. For example, a careful inspection of the above
constructions shows that the original construction of
\secref{long:hop:1:dim} augmented with the $2$-hop block spanner,
results in a graph with $O( (n/\prS) \log n)$ edges, and a hop
diameter $O( \log(1/\prS) )$.  Similarly, for any $k > 3$ integer, one
can construct a spanner that has $\leq \kk$-hop path, for all but
$n/\prS^2$ pairs, with
\begin{math}
    O\pth{ (n/\prS^{1+1/(\kk-1)}) \log n}
\end{math}
edges.

The construction is similar to the above.  To this end, let
$\nbl = 1/\prS^{1/(k-1)}$, and let
\begin{equation*}
    \sB
    =
    O\Bigl( \nbl \frac{\cKHop}{\prS} \ln n\Bigr)
    \qquad\text{and}\qquad%
    \sL
    =
    (\kk + 4) \sB.%
\end{equation*}
We connect all the pairs that are in distance $\leq \sL$ from each
other. Similarly, we inject the $2$-hop spanner on the blocks, where
every edge between two blocks is replaced by a random connector, where
the degree of a vertex in the connector, in expectation is
$O(\nbl/\prS)$. Let $\G$ be the resulting graph, and let
$\H \sim \DistrY{\H}{\prS}$.

As before, if a pair is long, then there is a path
$i \rightarrow B_1 \rightarrow B_2 \rightarrow \cdots \rightarrow
B_{\kk-1} \rightarrow j$. connecting them. Let $S_t$ all the vertices
of $B_t$ that are reachable via a $t$-hop path from $i$ in $\H$, and
let $\alpha_t = \cardin{S_t}$.  It is not hard to prove, using the
same argument as above, that with high probability, that
$ \alpha_1 \geq \prS \sB/2 =\Theta( \nbl \ln n) $, and more generally,
for $i \geq 1$, we have $\alpha_i = \Omega( \nbl^i \log n )$. Thus,
$\alpha_{\kk-1} = \Omega( \nbl^{\kk-1} \log n ) =
\Omega(\tfrac{1}{\prS}\log n)$. But then, with high probability there
is a direct edge from a vertex of $S_{\kk-1}$ to $j$. Thus, with high
probability, there is $\kk$-hop path from $i$ to $j$.

\begin{theorem}\RefProofInAppendix{k:hop} %
    \thmlab{k:hop}%
    Let $n > 0, k >3$ be two integers, and let $\prS \in (0,1)$ a
    parameter. The above constructed graph $\G$ over $\IRX{n}$ has
    $O( (n/\prS^{1+1/(\kk-1)}) \log n)$ edges. Furthermore, for a
    graph $\H \sim \DistrY{\G}{\prS}$, we have that $\H$ provides a
    $\leq \kk$-hop path for all pairs except (in expectation)
    $O( \kk n/\prS^{1+1/(\kk-1)} )$ pairs. For all pairs $i,j$, with
    $j > i + \Omega( \prS^{-1-1/(\kk+1)} \log n )$, such a $\kk$-hop
    straight path exists with high probability.
\end{theorem}
\begin{proof:e}{\Xthmref{k:hop}}{k:hop}
    We only need to bound the expected number of pairs that are short
    and fail to have a $\kk$-hop path -- their distance is smaller
    than $O( (1/\prS^{1+1/(\kk-1)}) \log n)$ (all other claims are
    implied by the above analysis). This follows by proving a bound on
    the $\kk$-deficiency of $L_n$, which we done below. See
    \lemref{l:b:k}.
\end{proof:e}

\subsubsection{Bounding the \TPDF{$\kk$}{k}-hop deficiency of the
   clique}

We need the following straightforward extension of \lemref{success}
(the proof is essentially the same, so we omit it).

\begin{corollary}
    \corlab{success:k}%
    Consider breaking $\IRX{n}$ into blocks of size $\xi c / \prS$,
    where $\xi > 1$ is an integer, and $c$ is the constant from
    \Eqref{c:value}.  Consider two blocks $B$ and $B'$, where $B$
    appears before $B'$, and $\gC_{i-1} = \goodY{B}{i-1} > 0$ be all
    the vertices in $B$ that are reachable from $1$ by a straight path
    with $\leq i-1$ hops. We have:
    \begin{compactenumA}
        \smallskip%
        \item
        $\Prob{\bigl. \goodY{B'}{i} \geq \min( 2 \xi \gC_{i-1}, \sB/3
           ) } \geq 1/2$.

        \smallskip%
        \item If $\gC_{i-1} \geq \sB/3$, then
        $\Prob{\bigl. \goodY{B'}{i} \geq (2/3)\sB } \geq 1/2$.
    \end{compactenumA}
\end{corollary}

\begin{lemma}\RefProofInAppendix{l:b:k} %
    \lemlab{l:b:k}%
    We have $\lbZ{n}{\prS}{\kk} \leq O(\kk n/\prS^{1+1/(\kk-1)} )$,
    see \defref{deficiency}.
\end{lemma}
\begin{proof:e}{\Xlemref{l:b:k}}{l:b:k}
    We are going to bound the expected number of vertices in
    $\G \sim \DistrY{K_n}{\prS}$ that do not have a $\kk$-hop path
    from $1$. To this end, we break $\IRX{n}$ into blocks of size
    $\sB = \xi(c/\prS)$, where $\xi = \ceil{ 4/ \prS^{1/(\kk-1)}}$,
    and let $\BSet = \{ B_1, \ldots, B_{n/\sB}\}$ be this set of
    blocks. Let $X_1$ be the minimum index, such that
    $ \goodY{B_{X_1}}{1} > \xi$. Since $\Ex{\goodY{B_i}{1}} = c \xi$
    and $c > 2$, this implies that $\Ex{X_1} \leq 2$. More generally,
    for $t > 1$, let $X_t$ be the minimal index such that
    \begin{math}
        \goodY{B_{X_t}}{t} > \xi^t.
    \end{math}
    By \corref{success:k}, we have that $\Ex{X_t - X_{t-1}} \leq 2$.
    In particular, $\Ex{X_{\kk-1} } \leq 2 \kk$. Observe, that any
    vertex in a block $B_j$, with $j > X_{\kk-1} + \Delta$ has at
    least (in expectation) $\Delta/2$ blocks, where each one of them
    has at least $\xi^{\kk-1} > 4/\prS$ vertices that are
    $(\kk-1)$-reachable. Indeed, each block $B_s$, for
    $s \in \IRY{X_{\kk-1}}{X_{\kk-1}+ \Delta}$, has at least
    probability half of having $4/\prS$ vertices that are
    $(\kk-1)$-reachable, and these events are independent (we are
    using here \corref{success:k} on each one of these blocks). Thus,
    using Chernoff's inequality, there is some constant $c_8$, such
    that the probability there are not at least $\Delta/4$ such good
    blocks, for $B_j$ is at most $\exp( -\Delta/c_8)$. And
    furthermore, if $B_j$ is good in this sense, then the probability
    that a vertex $u \in B_j$ does not have a $\kk$-hop path to it is
    at most
    \NotSoCG{%
       \begin{equation*}
           (1-\prS)^{(4/\prS)(\Delta/4)}
           \leq
           \exp( -\Delta).
       \end{equation*}
    }%
    \SoCG{%
       \begin{math}
           (1-\prS)^{(4/\prS)(\Delta/4)}
           \leq
           \exp( -\Delta).
       \end{math}
    }

    In particular, we have that the expected number of vertices of
    $\IRX{n}$ that are unreachable by a $\kk$-hop straight path from
    $1$ is bounded by
    \begin{equation*}
        \Ex{X_{\kk-1} \sB} + \sum_{\Delta=1}^\infty \bigl(\exp(-\Delta/c_8) +
        \exp(-\Delta)\bigr) \sB
        =%
        O( \kk \sB).
        \SoCG{\qedhere}%
    \end{equation*}
\end{proof:e}

\section{A \dependable spanner in \TPDF{$\Re^d$}{Rd}}
\seclab{main}

We need to use \LSO{}s, see \defref{local}, and in particular, the
following result of \NotSoCG{Chan \etal }\cite{chj-lota-20} for
computing a universal set of \LSO{}s.

\begin{theorem}[\tcite{chj-lota-20}] %
    \thmlab{lso}%
    For $\epsA \in (0,1/2]$, there is a set $\ordAll$ of
    $O(\log (1/\epsA)/\epsA^d)$ orderings of $[0,1)^d$, such that for
    any two points $\p, \q \in [0,1)^d$ there is an ordering
    $\order \in \ordAll$ defined over $[0,1)^d$, such that for any
    point $\pC$ with $\p \prec_\order \pC \prec_\order \q$ it holds
    that either $\dY{\p}{\pC} \leq \epsA \dY{\p}{\q}$ or
    $\dY{\q}{\pC} \leq \epsA \dY{\p}{\q}$ (i.e., $\order$ is
    $\eps$-local for $\pA$ and $\pB$).

    Furthermore, given such an ordering $\order$, and two points
    $\p,\q$, one can compute their ordering, according to $\order$,
    using $O(d\logeps)$ arithmetic and bitwise-logical operations.
\end{theorem}

\begin{theorem}%
    \thmlab{main}%
    Let $\P$ be a set of $n$ points in $\Re^d$, and let
    $\prS, \eps \in (0,1)$ be parameters.  One can construct a graph
    $\G$ over $\P$ with
    \begin{math}
        O\Bigl( \frac{C_\eps}{\prS^{4/3}} n \log n \Bigr)
    \end{math}
    edges, such that for all $\pA,\pB \in \P$, except maybe (in
    expectation) $O( C_\eps C_\prS n )$ pairs, we have that
    $\H \sim \DistrY{\G}{\prS}$ provides a $4$-hop path connecting
    $\pA$ and $\pB$, of length at most $(1+\eps)\dY{\pA}{\pB}$, where
    $C_\eps = O(\eps^{-d} \log \eps^{-1})$ and
    $C_\prS = O( \prS^{-4/3})$.
\end{theorem}
\begin{proof}%
    We compute the set $\ordAll$ of $\eps/8$-\LSO{}s provided by
    \thmref{lso} for $\P$, where $C_\eps = \cardin{\ordAll}$. For each
    \LSO $\order \in \ordAll$, construct the graph $\G_\order$ of
    \thmref{k:hop} (with $\kk=4$) over the points of $\P$, and let
    $\G$ be the union of all these graphs.

    Now, for any $\pA, \pB \in \P$, let $\order \in \ordAll$ be their
    $\eps/8$-local order. In expectation, except for $O(n/\prS^{4/3})$
    pairs, all other pairs have a $4$-hop path in
    $\G_{\order} \sim \DistrY{\G_\order}{\prS}$.  Assuming $\pA$ and
    $\pB$ have such a path
    $\pi \equiv \pA \rightarrow p_1 \rightarrow p_2 \rightarrow
    p_3\rightarrow \pB$ in $\G_\order$.  Let $\ell = \dY{\prS}{q}$.
    Observe that all the edges in this path, except exactly one
    segment, are either in $\ballY{\prS}{\ell \eps/8}$ or
    $\ballY{\prS}{\ell \eps/8}$. The total length of these short edges
    of $\pi$ is thus
    \begin{equation*}
        \leq 3 \cdot 2 \ell\eps/8 \leq (3/4)\eps \ell.
    \end{equation*}
    The single long edge across the two balls in the path has length
    $\leq \ell + 2(\eps \ell/8)$. Thus, the total length of
    \NotSoCG{the path }$\pi$ is at most
    $\ell + \eps \ell /4 + (3/4)\eps \ell= (1+\eps)\ell$, which
    implies the claim. ~
\end{proof}

An alternative approach is to set $\kk = \ceil{\log(1/\prS)}$, use
$\eps/(2\kk)$-\LSO{}s, and plug it into the above machinery. This
leads to the following.

\begin{corollary}
    \corlab{main}%
    Under the settings of \thmref{main}. one can construct a graph
    $\G$ over $\P$ with
    \begin{math}
        O\bigl( C \cdot \frac{1}{\eps^d \prS} \cdot n \log n \bigr)
    \end{math}
    edges, with
    $C = O \bigl( \log^d \tfrac{1}{\prS} \log \frac{\log
       (1/\prS)}{\eps} \bigr)$, such that for all $\pA,\pB \in \P$,
    except maybe (in expectation)
    \begin{math}
        O\bigl( n \, C \frac{1}{\eps^d \prS} \log \tfrac{1}{\prS}
        \bigr)
    \end{math}
    pairs, we have that $\H \sim \DistrY{\G}{\prS}$ provides a
    $2 \ceil{\log\tfrac{1}{\prS}}$-hop path connecting $\pA$ and
    $\pB$, of length at most $(1+\eps)\dY{\pA}{\pB}$.
\end{corollary}

\section{Conclusions}
\seclab{conclusions}

We leave many open problems for further research. First issue (but
arguably not that exciting) is finetuning the parameters -- can the
dependable spanner construction dependency be improved to
$1/\eps^{d-1}$ instead of $1/\eps^d$ (ignoring \si{polylogs}). The
recent work of Gao and Har-Peled \cite{gh-nolso-24} suggests this
should be doable. The current construction should also apply in other
settings where \LSO{}s are known, but maybe one can do better in some
of these cases? In the same vein, can one improve the dependency on
$\prS$ in the dependable spanner construction?

A potentially more interesting problem is trying to extend the results
when the probability of failure for every pair of points is provided
explicitly. Can one compute a good dependable spanner in such a case
of near optimal size?

\paragraph{Acknowledgement.}

We thank the anonymous referees for their numerous insightful
comments. In particular, they exposed a (arguably minor) mistake in a
previous version of this writeup, which is now corrected.

\BibTexMode{%
   \NotSoCG{
      \bibliographystyle{alpha}%
   }%
   \SoCG{%
      \bibliographystyle{plainurl}%
   }

   \bibliography{dependable_spanner}%
}%
\BibLatexMode{\printbibliography}

\appendix

\section{Standard tools}

\begin{theorem}[Chernoff's inequality]
    \thmlab{chernoff}%
    Let $X_1, \ldots, X_n \in \{0,1\}$ be $n$ independent random
    variables, with
    \begin{math}
        p_i = \Prob{ \bigl. X_i = 1 },
    \end{math}
    for all $i$.  let $X = \sum_{i=1}^{n} X_i$, and
    $\mu = \Ex{\bigl. X} = \sum_i p_i$. For all $\delta \geq 0$, we
    have
    \begin{math}
        \Prob{\bigl. X < (1-\delta)\mu } < \exp\bigl(-\mu\delta^2/2
        \bigr).
    \end{math}
\end{theorem}

\InsertAppendixOfProofs

\end{document}